\documentclass{article}
\usepackage{spconf}
\usepackage{amsmath}
\usepackage{amsthm}
\usepackage{amsfonts}
\usepackage{amssymb}
\usepackage{mathtools}
\usepackage{graphicx}
\usepackage{color}
\usepackage{verbatim}
\usepackage{hyperref}
\usepackage{todonotes}
\usepackage{enumitem}
\usepackage{cite}
\usepackage{cancel}

\newcommand{\white}{\color{white}}

\newcommand{\Ff}{\mathcal{F}}
\newcommand{\I}{\mathcal{I}}

\newcommand{\wb}{\bold{w}}

\newcommand{\Xb}{\bold{X}}

\newcommand{\At}{\tilde{A}}
\newcommand{\Bt}{\tilde{B}}
\newcommand{\Ct}{\tilde{C}}

\newcommand{\Rt}{\tilde{R}}

\newcommand{\wbt}{\tilde{\bold{w}}}
\newcommand{\Wbt}{\tilde{\bold{W}}}
\newcommand{\Scal}{\mathcal{S}}
\newcommand{\St}{\bar{\mathcal{S}}}
\newcommand{\Z}{\mathbb{Z}}
\newcommand{\R}{\mathbb{R}}

\newcommand{\E}{\mathbb{E}}
\newcommand{\F}{\mathbb{F}}
\newcommand{\N}{\mathbb{N}}
\newcommand{\gb}{\bold{g}}
\newcommand{\Sb}{\bold{S}}

\newcommand{\Bb}{\bold{B}}
\newcommand{\Bbt}{\tilde{\bold{B}}}

\newcommand{\Ib}{\bold{I}}
\newcommand{\ab}{\bold{a}}
\newcommand{\abt}{\tilde{\bold{a}}}

\newcommand{\eb}{\bold{e}}

\newcommand{\Var}{\mathrm{Var}}

\newtheorem{Thm}{Theorem}[section]

\newtheorem{Prop}[Thm]{Proposition}
\newtheorem{Rmk}[Thm]{Remark}

\title{Approximate Weighted $CR$ Coded Matrix Multiplication}

\name{$\text{Neophytos Charalambides}^{\dagger}$, $\text{Mert Pilanci}^{\ddagger}$, $\text{Alfred O. Hero III}^{\dagger},$\thanks{This work was partially supported by grant ARO W911NF-15-1-0479.}}
\address{$\text{\white .}^\dagger$EECS Department University of Michigan, $\text{\white .}^\ddagger$EE Department Stanford University}

\begin{document}
%
\maketitle
\begin{abstract}
One of the most common, but at the same time expensive operations in linear algebra, is multiplying two matrices $A$ and $B$. With the rapid development of machine learning and increases in data volume, performing fast matrix intensive multiplications has become a major hurdle. Two different approaches to overcoming this issue are, 1) to approximate the product; and 2) to perform the multiplication distributively. A \textit{$CR$-multiplication} is an approximation where columns and rows of $A$ and $B$ are respectively sampled with replacement. In the distributed setting, multiple workers perform matrix multiplication subtasks in parallel. Some of the workers may be stragglers, meaning they do not complete their task in time. We present a novel \textit{approximate weighted $CR$ coded matrix multiplication} scheme, that achieves improved performance for distributed matrix multiplication.
\end{abstract}

\begin{keywords}{Randomized numerical linear algebra, approximation algorithms, coded computing, coding theory.}
\end{keywords}

\vspace{-3mm}
\section{Introduction}
\label{intro}

\vspace{-2mm}
Matrix multiplication is one of the key underlying operations used in applications and algorithms in domains such as numerical analysis, machine learning, network analysis and scientific computing. Frequently, this operation occurs thousands of times, making it a a bottleneck and an impediment in large scale computations.

One way to speed up matrix multiplication, is to perform the necessary computations in a distributed manner, where a network of workers perform certain subtasks in parallel. A common obstacle in such networks is the presence \textit{stragglers}: workers whose computed task may never be received, due to delay or outage \cite{LLPPR18}. Matrix-matrix multiplication has been studied in this context \cite{YMAA17,FJHDCG17,DFHJCG19,SHN19,YMAA20,YA20}. Though optimal, these methods require relatively high encoding and decoding complexities; which may introduce numerical instability, and pose restrictions on the way the matrices need to be partitioned; in order to complete the multiplication in a distributive environment.

The $CR$-multiplication scheme produces a low-rank approximate matrix product \cite{DK01,DKM06a,DKM06b}. It judiciously sub-samples pairs of columns and rows, respectively, from $A$ and $B$, to form the approximation $CR\approx AB$. By appropriately weighting $C$ and $R$, the number of sampled pairs can be significantly reduced, while attaining the same approximation.

In this paper we introduce a procedure for \textit{straggler-robust weighted $CR$-multiplication} in distributed computing environments. Instead of sampling pairs of individual columns and rows, we sample submatrices which are partitions of $A$ and $B$, across the columns and rows respectively.

The paper is organized as follows. In section \ref{CR_mult} we describe $CR$ matrix multiplication, and extend the approximation theory to sampling pairs of submatrices of $A,B$; instead of pairs of rows and columns of $A,B$. Then, we introduce weighted $CR$-multiplication, which is beneficial when the associated sampling distribution is non-uniform. Weighted $CR$-multiplication reduces the number of operations, while producing the same approximation. In section \ref{str_matr_mult} we describe the “straggler problem” and \textit{coded matrix multiplication} (CMM). We then show how gradient coding (GC) schemes can be used to devise \textit{weighted CMM} (WCMM) schemes, and also provide a WCMM scheme based on the ``MatDot'' CMM \cite{FJHDCG17,DFHJCG19}. Finally, in section \ref{exp_sect} we present experimental results.

Our contributions are: (i) we generalize $CR$ matrix multiplication to accommodate sampling pairs of submatrices, instead of vectors, (ii) we propose \textit{weighted} CMM schemes; derived from GC and MatDot CMM, (iii) we provide theory showing that weighted and unweighted block-wise $CR$-multiplication produce the same result, (iv) we incorporate of $CR$-multiplication into WCMM, (v) we demonstrate the power of the proposed scheme by experiments.

\vspace{-3mm}
\section{Block-wise $CR$-multiplication}
\label{CR_mult}

\vspace{-1.6mm}
Consider the two matrices $A\in\R^{L\times N}$ and $B\in\R^{N\times M}$, for which we want to approximate the product $AB$. It is known that the product may be approximated by sampling with replacement, where the row-column sampling probabilities are proportional to their Euclidean norms. That is, we sample with replacement $r$ pairs $(A^{(i)},B_{(i)})$ for $i\in\N_N\coloneqq\{1,\cdots,N\}$ and $r<N$ ($A^{(i)}$=$i^{th}$ column of $A$, and $B_{(i)}$=$i^{th}$ row of $B$), with probability
\vspace{-1.8mm}
$$ p_j\propto\|A^{(j)}\|_2^2\cdot\|B_{(j)}\|_2^2 \vspace{-1.8mm} $$
and sum a rescaling of the samples' outer-products:
\vspace{-1.8mm}
$$ AB \approx \frac{1}{r}\cdot \left(\sum_{j\in\Scal}\frac{1}{p_j}A^{(j)}B_{(j)}\right) = \sum_{j\in\Scal}\frac{A^{(j)}}{\sqrt{rp_j}}\cdot\frac{B_{(j)}}{\sqrt{rp_j}} \vspace{-1.8mm} $$
where $\Scal$ is the multiset consisting of indices (possibly repeated) of the sampled pairs, hence $|\Scal|=r$. We denote the corresponding ``compressed versions'' of the input matrices by $C\in\R^{L\times r}$ and $R\in\R^{r\times M}$ respectively. This approximation satisfies $\|AB-CR\|_F=O(\|A\|_F\|B\|_F/\sqrt{r})$. Further details on this algorithm may be found in \cite{DK01,DKM06a,DKM06b,Woo14,Mah16}.

It is desired to sample submatrices with potentially more than one row and show how the above algorithm can be so modified. First, we partition $A$ and $B$ into $K$ disjoint submatrices consisting of $\tau\coloneqq N/K$ columns and rows respectively, which we denote by $\At_i\in\R^{L\times \tau}$ and $\Bt_i\in\R^{\tau\times M}$. That is
\vspace{-1.8mm}
$$ A=\Big[\At_1 \ \cdots \ \At_K\Big] \quad \text{ and } \quad B=\Big[\Bt_1^T \ \cdots \ \Bt_K^T\Big]^T. \vspace{-1.8mm} $$

Next, we determine the appropriate sampling probabilities. For $z\in\N_K$ and $\Pr[z=j]=\Pi_j$, let $X\in\R^{L\times M}$ be a matrix random variable with $\Pr[X=\At_j\Bt_j/\Pi_j]=\Pi_j$. From here on, we will be considering the case where $t$ block pairs $(\At_i,\Bt_i)$ are sampled with replacement, for $t<K$.

For $\E[X]$ the statistical expectation, we have
\vspace{-1.8mm}
$$ \E[X] = \sum_{l=1}^K
\Pr[z=l]\cdot\frac{1}{\Pi_l}\At_l\Bt_l = \sum_{l=1}^K\At_l\Bt_l = AB \vspace{-1.8mm} $$
thus $X$ is an unbiased estimator. Furthermore, we have
\begin{align*}
\label{var_expr}
  \E\left[\|AB-X\|_F^2\right] &= \left(\sum_{l=1}^K\frac{\|\At_l\|_F^2\|\Bt_l\|_F^2}{\Pi_l}\right)-\|AB\|_F^2 \nonumber\\ &\eqqcolon f(\{\Pi_l\}_{l=1}^K) - \|AB\|_F^2 = \Var(X) \ . \vspace{-1.8mm}
\end{align*}
This also implies
\vspace{-1.8mm}
\begin{equation}
\label{expect_bound}
  \E\left[\|AB-X\|_F^2\right] \leq \|A\|_F^2\|B\|_F^2 \ . \vspace{-1.8mm}
\end{equation}

To define the optimal distribution $\{\Pi_i\}_{i=1}^K$, we minimize $f(\{\Pi_i\}_{i=1}^K)$ subject to $\sum_{i=1}^K\Pi_i=1$. Introduce the Lagrange multiplier $\lambda$:
$$ g(\{\Pi_i\}_{i=1}^K) \coloneqq f(\{\Pi_i\}_{i=1}^K) + \lambda\left(\sum_{i=1}^K\Pi_i-1\right) \vspace{-1.8mm} $$
for which
\vspace{-1.8mm}
$$ 0=\frac{\partial g}{\partial\Pi_i} = -\frac{1}{\Pi_i^2}\cdot\|\At_i\|_F^2\|\Bt_i\|_F^2 + \lambda $$
$$ \implies \quad \Pi_i \ = \ \|\At_i\|_F\|\Bt_i\|_F\big/\sqrt{\lambda} \ \propto \ \|\At_i\|_F\|\Bt_i\|_F \ $$
so it is natural to define the sampling distribution by
\vspace{-1.8mm}
\begin{equation}
\label{der_distr}  
  \Pi_i = \frac{\|\At_i\|_F\|\Bt_i\|_F}{\sum_{l=1}^n\|\At_l\|_F\|\Bt_l\|_F}\ .
\end{equation}

By the second derivative test; since $\partial^2g/\partial\Pi_i^2=\frac{2}{\Pi_i^3}\cdot\|\At_i\|_F^2\|\Bt_i\|_F^2>0$  for all $i\in\N_K$, the sampling probabilities \eqref{der_distr} are optimal. That is, they minimize the expectation of \eqref{expect_bound}, which is equal to the sum of the variances of all the elements of the matrix product.

To reduce the variance, we take $t$ independent copies of the estimator $X$, which we denote by $X_{\iota}$ for $\iota=1,\cdots,t$. Let $\Ct$ and $\Rt$ be the resulting sketches of $A$ and $B$ respectively, by averaging the $t$ samples which were drawn to obtain
\vspace{-1.8mm}
\begin{equation}
\label{appr_id}  
  Y=\Ct\Rt = \frac{1}{t}\sum_{\iota\in\St}\frac{\At_\iota\Bt_\iota}{\Pi_\iota} = 
\sum_{\iota\in\St}\frac{\At_\iota}{\sqrt{t\Pi_\iota}}\cdot\frac{\Bt_\iota}{\sqrt{t\Pi_\iota}}
\vspace{-1.8mm}
\end{equation}
for $\St$ the multiset of indices of the $t$ sampled pairs, and
\vspace{-2mm}
$$ \Ct=\frac{1}{\sqrt{t}}\begin{bmatrix}\At_{\St_1}\big/\sqrt{\Pi_{\St_1}}^{\white.} & \cdots & \At_{\St_t}\big/\sqrt{\Pi_{\St_t}}^{\white.}\end{bmatrix} \in \R^{L\times t\tau} $$
$$ \Rt=\frac{1}{\sqrt{t}}\begin{bmatrix} \Bt_{\St_1}^T\big/\sqrt{\Pi_{\St_1}}^{\white.} & \cdots & \Bt_{\St_t}^T\big/\sqrt{\Pi_{\St_t}}^{\white.} \end{bmatrix}^T  \in \R^{t\tau\times M} $$
hence the corresponding rescalings which take place in the block sampling scenario, are $1/\sqrt{t\Pi_\iota}$. Moreover
\vspace{-1.8mm}
\begin{equation}
\label{unb_expr} 
  \E[Y] = \E\Big[\sum_{\iota\in\St}X_{\iota}\Big] = \sum_{\iota\in\St}\E[X_{\iota}] = AB \qquad \vspace{-1.8mm}
\end{equation}
\vspace{-1.8mm}
$$ \Var(Y) = \Var\Big(\frac{1}{t}\sum_{\iota\in\St}X_{\iota}\Big) = \frac{1}{t}\Var(X)\leq\frac{1}{t}\|A\|_F^2\|B\|_F^2 \ . $$
\begin{Thm}
  The estimator $Y=\Ct\Rt$ is unbiased, while the sampling probabilities $\{\Pi_i\}_{i=1}^K$ minimize $\Var(Y)$, and the approximation satisfies $\|AB-\Ct\Rt\|_F^2=O\left(\|A\|_F^2\|B\|_F^2/t\right)$.
\end{Thm}

\begin{proof}
 We already established the first two claims. The third clain is a direct application of Jensen's inequality to \eqref{expect_bound}.
\end{proof}

By applying Markov's inequality, one can also get a concentration bound on the matrix product approximation.

We can define the sampling and compression in terms of a matrix $\Sb\in\R^{N\times r}$, using \eqref{appr_id}. First, we construct $S\in\R^{K\times t}$ by initializing it to $S=\bold{0}_{K\times t}$. Subsequently:
\begin{enumerate}[noitemsep,nolistsep]
  \item randomly sample independently with replacement based on $\{\Pi_i\}_{i=1}^K$; until $t$ \textit{distinct} blocks have been selected, and let $\St$ be the samples' index multiset
  \item assign $S_{_{\St_j,j}}=1/\sqrt{|\St|\Pi_j}$, for $j=1,\cdots, t$ ,
\end{enumerate}
and finally form $\Sb=S\otimes\Ib_{\tau}$. It is clear that
\vspace{-1.6mm}
$$ \Ct=A\cdot\Sb \ \text{ and } \ \Rt=\Sb^T\cdot B \ \ \implies \ \ AB\approx A\cdot(\Sb\Sb^T)\cdot B. $$

\vspace{-1.6mm}
Note that since the total number of samples of pairs we are considering has increased, we need to rescale by $\sqrt{|\St|\Pi_j}$, instead of $\sqrt{t\Pi_j}$.

\begin{Rmk}
The sampling can be done efficiently with an additional $O(1)$ additional storage space, by a modification of the pass-efficient $\textup{\textsf{SELECT}}$ algorithm form \cite{Mah16}.
\end{Rmk}

\vspace{-3mm}
\subsection{Weighted $CR$-multiplication}
\label{wht_CR_mult}

The above sampling procedure reduces the the size of each of the two matrices by a compression factor $\rho\coloneqq N/r$. This then reduces the overall operations of naive matrix multiplication from $O(LMN)$ to $O(LMr)$, i.e. it drops by a factor of $1/\rho$. We will take advantage of the potential non-uniformity of $\{\Pi_i\}_{i=1}^K$ to drop the operation count even more. The non-uniformity is also a consequence of potential outliers, which are common to machine leaning and real datasets.

For the \textit{weighted} variant of the proposed $\Sb$, we first construct $S_{\wb}\in\R^{K\times t}$ and $\wb$ as follows:
\begin{enumerate}[noitemsep,nolistsep]
  \item randomly sample independently with replacement based on $\{\Pi_i\}_{i=1}^K$ until $t$ \textit{distinct} blocks are drawn, and let $\St$ be the corresponding index multiset, and $\I$ the set of indices comprising $\St$, i.e. $\I\subsetneq\N_K$ has no repetitions and $|\I|=t$
  \item retain each sampled pair only once, and count how many times each pair was drawn --- these counts correspond to the entries of the \textit{weight vector} $\wb\in\N_0^{1\times K}$
  \item assign $(S_\wb)_{_{\I_j,j}}=\sqrt{\wb_{\I_j}}/\sqrt{|\St|\Pi_j}$, for $j=1,\cdots, t$, where each row of $S_{\wb}$ has up to one nonzero entry.
\end{enumerate}
We define the weighted compression matrix as $\Sb_{\wb}=S_{\wb}\otimes \Ib_\tau$, and the resulting sketches as $C_{\wb}=A\cdot\Sb_{\wb}$ and $R_{\wb}=\Sb_{\wb}^T\cdot B$.

One could sample a fixed number of times, as is done in most such algorithms. We describe the case where $t$ distinct pairs are sampled, to make the connection with distributed computations more natural; as was done in \cite{CPH20}.

\begin{Prop}
The resulting approximations from the algorithms using $\Sb$ and $\Sb_\wb$ respectively, are identical.
\end{Prop}

\begin{proof}
By our constructions, matrix $\Sb$ is defined by $\St$; and $\Sb_\wb$ is defined by $\I$ and $\wb$. Let $\wbt=\wb_{|_\I}\in\Z_+^{1\times t}$, for which $|\St|=\|\wb\|_1=\|\wbt\|_1$. By a simple computation
\vspace{-2mm}
\begin{align*}
  C_{\wb}\cdot R_{\wb} &= A\cdot(\Sb_\wb\Sb_\wb^T)\cdot B = A\cdot\left(\sum_{i\in\I}\wbt_i\cdot\frac{\eb_i\cdot\eb_i^T}{|\St|\Pi_i}\right)\cdot B\\
  &= A\cdot\left(\sum_{j\in\St}\frac{\eb_j\cdot\eb_j^T}{|\St|}\Pi_j\right)\cdot B = A\cdot(\Sb\Sb^T)\cdot B = \Ct\cdot\Rt
\end{align*}
for $\eb_i$ the standard basis column vectors of length $N$.
\end{proof}

The benefit of applying $\Sb_\wb$ in place of $\Sb$, is that $\Sb_\wb$ has a factor $\|\wbt\|_1/\|\wbt\|_0\geq1$ fewer rows than $\Sb$, so all in all we would need to store a factor $(\|\wbt\|_1/\|\wbt\|_0)^2$ fewer many matrix entries after the sampling takes place. The number of operations needed to carry out the product of the sketched matrices also drops by the same factor.

\vspace{-3mm}
\section{Stragglers and Coded Multiplication}
\label{str_matr_mult}

\vspace{-2mm}
\subsection{Straggler Problem}
\label{str_problem}

Consider a single central server node that has at its disposal the matrices $A$ and $B$, such that it can distribute submatrices of $A$ and $B$ among $n$ workers to compute the estimate $Y$. The two matrices can be complex, but we focus on the real valued case. One way to compute $AB$ is to use the identity
\vspace{-1.8mm}
\begin{equation}
\label{outer_prod_expr}
  AB=\sum_{i=1}^NA^{(i)}B_{(i)}=\sum_{l=1}^K\At_l\Bt_l
\vspace{-1.8mm}
\end{equation}
which we have already exploited in section \ref{CR_mult}. This makes the process parallelizable, allowing parallel incorporation of the block-wise $CR$ matrix multiplication. The central server determines the appropriate way to distribute the sampled pairs with a certain level of redundancy, in order to recover $Y$.

The goal is to recover a weighted linear combination of the rank-$\tau$ outer products $\At_j\Bt_j$ of the distinct sampled pairs:
\vspace{-3mm}
$$ Y_{\wbt}\coloneqq \underbrace{(\wbt\otimes\Ib_L)}_{\Wbt\in\R^{L\times tL}}\cdot\underbrace{\Big[X_{\I_1}^T\ | \ \cdots\ | \ X_{\I_t}^T\Big]^T}_{\Xb\in\R^{tL\times M}} = \sum_{j=1}^t\wbt_\iota\cdot X_\iota = \sum_{j=1}^t Y_j \vspace{-3mm} $$
for $Y_j\coloneqq\wbt_j\cdot\At_{\I_j}\Bt_{\I_j}=\wbt_j\cdot X_{\I_j}\in\R^{L\times M}$.

In the distributed setting each worker node completes its task by sending back a \textit{weighted partial sum}, i.e. a weighted linear combination of a subset of the subtask products $\At_i\Bt_i$. Different types of failures can occur during the computation or the communication process. These failures are what we refer to as \textit{stragglers}, that are ignored by the main server: specifically, the server only receives $f\coloneqq n-s$ completed tasks. Here $s$ is the number of stragglers our scheme can tolerate and $f$ is referred to as the \textit{recovery threshold}. We denote by $\Ff\subsetneq\N_n$ the index set of the $f$ fastest workers who complete their task. Once \textit{any} set of $f$ tasks is received, the central server should be able to decode the received encoded subtasks, and therefore recover the approximation $Y_{\wbt}$.

We note that this kind of weighted multiplication can be utilized in other applications, such as covariance and Hessian estimation, as well as in rank tree tensors \cite{LH99,GNC18}.

\vspace{-3mm}
\subsection{Coded Matrix Multiplication}
\label{CMM}

Coded matrix multiplication (CMM) \cite{LSR17} is a principled framework for providing redundancy in centralized distributed computing networks. CMM guarantees recovery in the presence of stragglers, as long as the number of stragglers does not exceed an upper bound. Each worker is asked to perform a computation and encode the result, before sending it to the main server. We show how to recover the weighted sum of the rank-$\tau$ outer products; from the encoded results, which complements the weighted $CR$-multiplication.

First, we show how to devise WCMM schemes from GC schemes. The robustness characteristics match those of the GC schemes they are based off. Then, we propose another scheme based on the ``MatDot'' CMM scheme from \cite{FJHDCG17,DFHJCG19}.

\vspace{-3mm}
\subsection{Weighted CMM Schemes from Gradient Coding}
\label{1st_WCMM}

Gradient coding (GC), first proposed in \cite{TLDK17}, is a technique for straggler mitigation in distributed learning. It consists of an encoding matrix $\Bb\in\Sigma^{n\times t}$ and a decoding vector $\ab_{\Ff}\in\Sigma^n$ for each possible $\Ff$, which satisfy $\ab_{\Ff}^T\Bb=\bold{1}_{1\times t}$, for some alphabet $\Sigma$. We leverage \eqref{outer_prod_expr} in order to construct weighted CMM schemes from GC schemes such as \cite{HASH17,RTTD17,OGU19,CMH20}.

\begin{Thm}
Any GC scheme which accommodates multiplication with real numbers and constructs $\ab_{\Ff}$ online once $\Ff$ is known, can be turned into a weighted CMM.
\end{Thm}

\begin{proof}
Considering any such scheme $(\ab_{\Ff},\Bb)$ for any non-straggler index set $\Ff$, observe that
\begin{align*}
  (\overbrace{\ab_{\Ff}^T\otimes\Ib_L}^{\abt_{\Ff}^T\in\Sigma^{L\times nL}})&\cdot(\overbrace{\Bb\cdot diag(\wbt)\otimes \Ib_L}^{\Bbt\in\Sigma^{nL\times tL}})\cdot\Xb = \\
  &= (\ab_{\Ff}^T\cdot\Bb\cdot diag(\wbt))\otimes(\Ib_L\cdot\Ib_L)\cdot\Xb\\
  &= \big((\bold{1}_{1\times t}\cdot diag(\wbt))\otimes\Ib_L\big)\cdot\Xb\\
  &= (\wb\otimes\Ib_L)\cdot\Xb = \sum_{j=1}^t\wbt_j\cdot X_j \ . \vspace{-2mm}
\end{align*}
Hence, $(\abt_{\Ff},\Bbt)$ is a weighted CMM scheme.
\end{proof}

\begin{Prop}
\label{prop_GC}
By compressing matrices $A,B$ by a factor of $\rho$ along the appropriate dimension, while not reducing the workload of the workers, we can now tolerate $\grave{s}=\rho(s+1)-1$ stragglers by modifying $(\abt_{\Ff},\Bbt)$, for $\rho$ such that $0<\grave{s}<n$.
\end{Prop}

By reducing the respective dimensions of $A$ and $B$ by $\rho$, we have a new encoding matrix $\grave{\Bb}$ of the same size as $\Bb$, whose columns correspond to blocks of $(1/\rho)$-times the size of what was previously considered, hence the workers can now be allocated $\rho$-times as many blocks, i.e. $\|\grave{\Bb}_{(i)}\|_0=\rho\cdot\|\Bb_{(i)}\|_0$. By the construction of most GC schemes, this would result in $\rho$-times as many workers being allocated the same subset of blocks; compared to the initial scheme. That is, $\grave{s}+1=\|\grave{\Bb}^{(j)}\|_0=\rho\cdot\|\Bb^{(j)}\|_0=\rho\cdot(s+1)$, thus $\grave{s}=\rho\cdot(s+1)-1$. In essence, considering the GC from \cite{CMH20}, we now work with congruence classes $\bmod(\rho(s+1))$ as opposed to $\bmod(s+1)$, which implies we require fewer workers to respond. Analogous arguments hold for other GC schemes.

\vspace{-2mm}
\subsection{A Second Weighted CMM Scheme}
\label{2nd_WCMM}

Another scheme we present is a byproduct of the MatDot CMM scheme \cite{FJHDCG17,DFHJCG19} that exploits \eqref{outer_prod_expr}. In the MatDot scheme, an evaluation encoding polynomial of the submatrices $\At_i$ and $\Bt_i$ takes place for $p_{A}(x)=\sum_{j=1}^{t}\At_{j}x^{j-1}$ and $p_{B}(x)=\sum_{j=1}^{t}\Bt_{j}x^{t-j}$;
over arbitrary distinct elements $x_1,\cdots,x_n$ of a finite field $\F_q$ for $q>n$. The $i^{th}$ worker receives the encodings $p_A(x_i)$ and $p_B(x_i)$, i.e. the evaluation of the encoding polynomials at the evaluation point corresponding to the worker. Each worker is requested to communicate the computation $C(x_i)=p_A(x_i)\cdot p_B(x_i)$, which is a polynomial of degree $2(t-2)$. The sum of all the outer-products is the coefficient of $x^{t-1}$ of the polynomial $p_A(x)\cdot p_B(x)$. Once any $2t-1$ evaluations of the polynomial $C(x)$ on distinct points are received, polynomial interpolation of Reed-Solomon decoding can be applied in order to retrieve the product $AB$.

To incorporate the weights into the MatDot scheme, we use the polynomials $\tilde{p}_{A}(x)=\sum_{j=1}^{t}\sqrt{\wbt_j}\cdot\At_{j}x^{j-1}$ and $\tilde{p}_{B}(x)=\sum_{j=1}^{t}\sqrt{\wbt_j}\cdot\Bt_{j}x^{t-j}$ for the encoding, and the workers carry out the computations $C_{\wbt}(x_i)=\tilde{p}_A(x_i)\cdot \tilde{p}_B(x_i)$. The decoding step remains same.

\begin{Prop}
\label{prop_matdot}
Under the same assumptions of Proposition \ref{prop_GC}, the recovery threshold of our second WCMM drops from $2t-1$ to $2\grave{t}-1=2(t/\rho)-1$, for $\rho\mid t$. 
\end{Prop}

Considering $A,B$, once the sampling takes place we only deal with a total of $\grave{t}=t/\rho$ block pairs. Equivalently, the polynomials corresponding to each worker are now defined by $\rho$-many weighted  ``compressed'' pairs. Thus, we now only need to wait for $2\grave{t}-1$ responses to perform the interpolation.

\section{Experiment}
\label{exp_sect}

\vspace{-2mm}
\subsection{Minimum Variance of Frobenius Error}
\label{min_var_experiment}

We compare our weighted $CR$ approximation, to a $CR$ approximation with uniform sampling. We construct random matrices $A,B$ with $L=260$, $N=9600$, $M=280$; with non-uniform distribution $\{\Pi_i\}_{i=1}^K$ and $\|A\|_F^2\|B\|_F^2=O(10^{11})$. The minimum benefit of our sampling approach, occurs when $\{\Pi_i\}_{i=1}^K$ is close to uniform. We ran ten different instances for compression factors $\rho=K/t$ between $2$ and $16$. We kept $K=480$, $\tau=20$ fixed and varied $t$. On the plot we indicate the average approximation error $\|AB-C_{\wb}R_{\wb}\|_F^2$ along with the variance, and the corresponding error and variance for the uniform sampling approach. It is evident that our scheme has minimal variance over all values of $\rho$, and the error increases slightly as this factor decreases.
\begin{figure}[h]
  \centering
  \label{comp_times}
    \includegraphics[scale=.11]{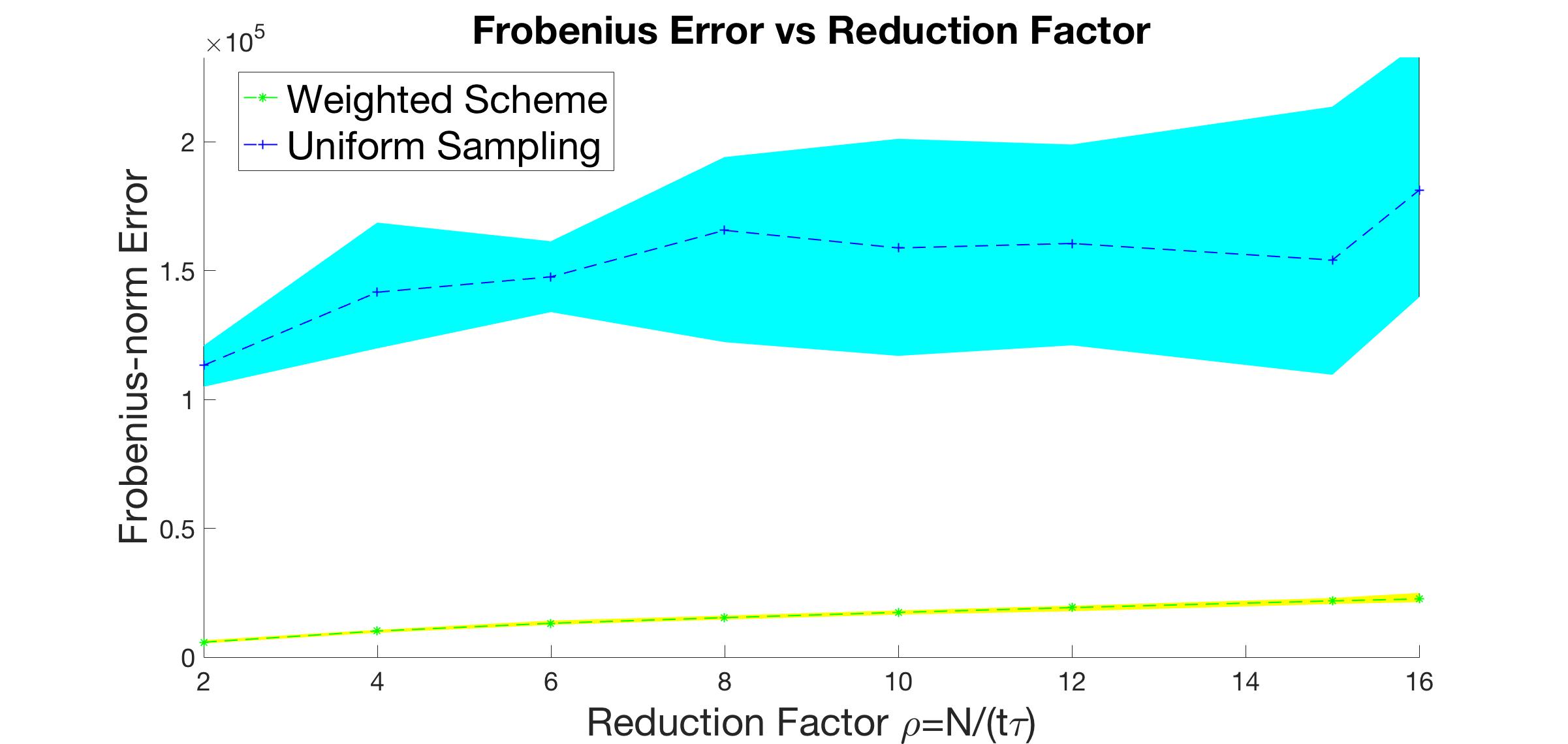}
    \caption{Average errors and variance, over ten simulations}
\end{figure}

\vspace{-3mm}
\subsection{Straggler Tolerance with AWS Job-Times}

We carried out the $CR$ approximation with identical parameters as in section \ref{min_var_experiment}, except for $N=10^4$ and $K=500$. We implemented the scheme of section \ref{1st_WCMM}, with $n=500$ and $s=19$, with worker completion times taken from 500 AWS-servers completing a job \cite{BP19}. For a compression factor $\rho=20$ we were able to compute the approximation in 10\% the time required by the corresponding exact recovery scheme, as we tolerated $\grave{s}=399$ stragglers. The approximation had a relative error of $\frac{\|AB-C_{\wb}R_{\wb}\|_F^2}{\|A\|_F^2\|B\|_F^2}=8.26\times 10^{-7}$. The total time the exact method needed to accommodate 19 stragglers, corresponded to a reduction of $\rho=2$, for which our scheme had a relative error of $1.92\times 10^{-7}$.



\bibliographystyle{unsrt}
\bibliography{refs.bib}

\end{document}